\documentclass[12pt]{article}
\usepackage{amsmath,amsfonts,amsthm,amssymb,color}
\usepackage{fancybox}

\newtheorem{theorem}{Theorem}[section]
\newtheorem{corollary}[theorem]{Corollary}
\newtheorem{lemma}[theorem]{Lemma}

\theoremstyle{Definition}
\newtheorem{Definition}[theorem]{Definition}

\newenvironment{fminipage}%
  {\begin{Sbox}\begin{minipage}}%
  {\end{minipage}\end{Sbox}\fbox{\TheSbox}}

\newenvironment{algbox}[0]{\vskip 0.2in
\noindent 
\begin{fminipage}{6.3in}
}{
\end{fminipage}
\vskip 0.2in
}

\def\norm#1{\left\| #1 \right\|}

\newcommand\bb{\boldsymbol{\mathit{b}}}

\newcommand\ff{\boldsymbol{\mathit{f}}}

\newcommand\uu{\boldsymbol{\mathit{u}}}

\newcommand\RR{\boldsymbol{\mathit{R}}}

\newcommand{\polylog}{\text{polylog}}
\newcommand{\etal}{~\textit{et. al.}}

\begin{document}

\title{
Approximate Undirected Maximum Flows in\\$O(m\polylog(n))$ Time
}

\author{
Richard Peng
\\Georgia Tech\\rpeng@cc.gatech.edu
}

\maketitle

\begin{abstract}
We give the first $O(m\polylog(n))$ time algorithms for approximating
maximum flows in undirected graphs and constructing
$\polylog(n)$-quality cut-approximating hierarchical tree decompositions.
Our algorithm invokes existing algorithms for these two problems
recursively while gradually incorporating size reductions.
These size reductions are in turn obtained via ultra-sparsifiers,
which are key tools in solvers for symmetric diagonally
dominant (SDD) linear systems.
\end{abstract}

\section{Introduction}

The problem of finding maximum flows and minimum cuts has been studied
extensively in algorithmic graph theory and combinatorial optimization.
It led to important tools in algorithm design such as
augmenting paths~\cite{FordF56},
blocking flows~\cite{Dinic70,EdmondsK72},
dynamic trees~\cite{GalilN80,SleatorT83},
dual algorithms~\cite{GoldbergT86},
scaling algorithms~\cite{GoldbergR98},
graph sparsification~\cite{BenczurK96},
and electrical flows~\cite{ChristianoKMST10,Madry13}.
In its simplest form, the maximum flow problem asks to route the most
flow from a source to a sink while obeying edge capacities.
Its dual, the minimum cut problem, asks for the minimum capacity
of edges whose removal disconnects the sink from the source.

Approximating maximum flows in undirected graphs
has received much attention recently due to its tighter interactions
with randomized and numerical tools~\cite{ChristianoKMST10,LeeRS13,Sherman13,KelnerLOS14}.
Algorithms for this variant have applications in graph
partitioning~\cite{KhandekarRV06:journal,OrecchiaSVV08,Sherman09},
image processing~\cite{ChinMMP13}, and as we will describe, the construction
of oblivious routing schemes~\cite{RackeST14}.

Recently, algorithms that approximate undirected maximum flows in
$O(m^{1 + o(1)} \epsilon^{-2})$ time were given by
Sherman~\cite{Sherman13} and  Kelner\etal~\cite{KelnerLOS14}.
At the core of these algorithms are congestion-approximators~\cite{Madry10} and
oblivious routing schemes respectively~\cite{KelnerLOS14}.
Congestion-approximators can be viewed as a small set of representative
cuts in the graph, and oblivious routing schemes are more powerful in that
they preserve flows as well as cuts.
The runtime of these algorithms stems from both the quality of these
approximators as well as the cost of constructing them.
A natural question stemming from them is to further
improve this running time.

Oblivious routing schemes are of independent interest in the study
of graph partitioning and routing.
Schemes with quality $\polylog(n)$ were shown to exist by R\"{a}cke~\cite{Racke02},
and invoking them would lead to a better running time of $m\polylog(n)$
after preprocessing.
However, finding these schemes requires solving an intricate sequence of
ratio cut problems~\cite{HarrelsonHR03,BienkowskiKR03}.

The current best algorithms for approximating ratio cuts are based on invoking
(approximate) maximum flows~\cite{KhandekarRV06:journal,OrecchiaSVV08,Sherman09}.
Following the break-through on approximate maximum flows,
R\"{a}cke \textit{et al.}~\cite{RackeST14} gave a more efficient algorithm
for constructing oblivious routing schemes.
This result can be viewed as producing a  $\polylog(n)$-quality oblivious
routing scheme by computing maximum flows on graphs of total size
$O(m\polylog(n))$.
This leads to a chick-and-egg situation when approximators and maximum flow
algorithms are viewed as black boxes:
either gives the other via an overhead of $\polylog(n)$, but to get the calls
started we need to invoke routines that run in $O(m^{1 + o(1)})$ time and
produce $m^{o(1)}$-approximations.

In this paper, we complete this cycle of algorithmic invocations by resolving this
chicken-and-egg situation, leading to improved algorithms to all intermediate problems.
The key observation is that the oblivious routing schemes produced by
the R\"{a}cke~\textit{et al.} algorithm have fixed size:
producing them via recursive calls does not affect the cost of invoking them,
and any error introduced in the recursion will only show up
as a slightly larger overhead on this fixed size.
The main steps of our algorithm on a graph $G$ are:
\begin{enumerate}
\item Produce a graph $H$ with size $m / \polylog(n)$
that can $\polylog(n)$-approximate $G$.
\item Construct an approximator for $H$ using the R\"{a}cke\etal~
algorithm, making more recursive maximum flow calls.
\item Convert this scheme to one for $G$,
and use it to solve approximate maximum flows.
\end{enumerate}

The size reduction allows us to bound the total size of the maximum flow
instances computed recursively by at most $m / 2$, giving a total size bound
of $O(m)$.
As $H$ $\polylog(n)$-approximates $G$,
the approximator for $H$ returned by the recursive calls is still a
$\polylog(n)$-quality approximator for $G$.
The fact that its size is $O(n)$ then allows us to bound the overall
cost by $O(m\polylog(n))$.

This recursive scheme allows us to bypass the more expensive
approximators used to initiate this sequence of algorithmic calls.
The total cost in turn reduces from $O(m^{1 + o(1)})$ to $O(m\polylog(n))$.
Furthermore, these size reductions can be directly obtained via ultra-sparsifiers
from solvers for linear systems in graph Laplacians~\cite{SpielmanTengSolver:journal}.
This results in a short pseudocode when the pieces are viewed as black-boxes.
Our algorithm is also analogous to iterative schemes for
computing row samples of matrices~\cite{LiMP13}: the congestion-approximator plays
a similar role to the small row sample, and the call structure is analogous
to what we use here, with ultra-sparsifiers being the size reductions.

We will introduce the algorithmic tools that we invoke
in Section~\ref{sec:background},
and describe our algorithm in Section~\ref{sec:algo}.
For simplicity, we will limit our presentation to the cut setting
and utilize the oblivious routing schemes as congestion-approximators.
As the oblivious routing construction by R\"{a}cke~\textit{et al.}~\cite{RackeST14}
also produces embeddings, and size reductions similar to ultra-sparsifiers
were used in the flow based algorithm by Kelner~\textit{et al.},
we believe this scheme can be extended to the flow setting as well.

We will also not optimize for the exponent in $\log{n}$ because
further runtime improvements based on this approach are likely.
However, major obstacles remain in obtaining running times of $m\log^{5}n$
or faster:
\begin{enumerate}
\item Random sampling based ultra-sparsifiers incur an overhead of $\log^2{n}$ in error.
\item Current oblivious routing constructions are based on top-down
divide-and-conquer with $\log{n}$ levels, each making a sequence of $\log{n}$
maximum flow calls through rebalancings~\cite{RackeST14}.
At present these routines also incur several additional log factors due to error
accumulations over levels of recursion.
\item Oblivious routing schemes incur a distortion of at least $\log{n}$~\cite{Racke08}.
\item Producing balanced cuts using maximum flows requires $\log{n}$
maximum flow invocation~\cite{OrecchiaSVV08,Sherman09}.
\item The invocation of congestion-approximators to produce approximate maximum flows
requires an iteration count that's at least quadratic in the distortion,
as well as incurring another $\log{n}$ factor overhead.~\cite{Sherman13,KelnerLOS14}.
\footnote{We cite the lower cost bounds from~\cite{KelnerLOS14} here under the
belief that these routines have inherent connections.}
\end{enumerate}
Directly combining these estimates leads to a total cost of about $m\log^{11}n$.
An optimistic view is that the algorithms using congestion-approximators can depend
linearly on the distortion, and reusing maximum flow calls across the construction
scheme leads to recursion on graphs with total size $m\log{n}$.
Even in this case, the overall cost is still about $m \log^5{n}$.
Therefore, we believe obtaining a running time of $O(m\log^3{n})$
will require significant improvements to both algorithms that construct
oblivious routings and iterative methods that utilize them.

\section{Background}
\label{sec:background}
Our presentation follows the notations from~\cite{Sherman13}
and~\cite{KelnerLOS14}.
A flow $\ff$ meets demands $\bb$ if for all vertices $v$,
the total amount of flow enter/leaving $v$ is $\bb_v$.
For edge capacities $\uu$, the congestion of $\ff$
is the maximum of $|\ff_e / \uu_e|$ over all edges.
By a standard reduction via binary search (e.g. Sections 2.2 and 3.1 of
~\cite{ChristianoKMST10}), we can focus on the decision version.
For a fixed demand, the problem asks to either route it with congestion at
most $1 + \epsilon$, or certify via a cut that it cannot be routed
with congestion less than $1$.

A cut is defined by a subset of vertices $S$:
its demand, $\bb(S)$, is the total demand of vertices in $S$,
and its capacity, $\uu(S)$, is the total capacity of edges leaving $S$.
The ratio between demand and capacity is a lower bound
for the minimum congestion, and the maxflow-mincut theorem
states that the minimum congestion needed to route a demand
is in fact equal to the maximum demand/capacity ratio over all cuts $S$.

The connections between flows, cuts, congestion, and demand brings us
to the notion of $(1 + \epsilon)$-approximate flow/cut solutions, which
will be our standard notion of approximate solutions.
For a demand $\bb$ and an error $\epsilon$, such a pair consists of a flow
and cut whose congestion and demand/capacity value
are within a factor of $1 + \epsilon$ of each other.

We will make extensive use of approximations, and denote them
using the $\approx_{\kappa}$ notation.
For two scalar quantities, $x$ and $y$, we use
$x \approx_{\kappa} y$ to mean that there exist parameters
$\gamma_{\min}$ and $\gamma_{\max} \leq \gamma_{\min} \kappa$
such that $\gamma_{\min} x \leq y \leq \gamma_{\max} x$.


\subsection{Congestion Approximators}

We will use an algorithm by Sherman~\cite{Sherman13} on using
 congestion-approximators to compute approximate maximum flows.
 
 \begin{Definition}[Definition 1.1. in~\cite{Sherman13:arxiv}]
An $\alpha$-congestion-approximator of $G$ is a matrix $\RR$ such that
for any demand vector $\bb$,
\[
\norm{\RR \bb}_{\infty} \approx_{\alpha} \text{opt}(\bb)
\]
where $\text{opt}(\bb)$ is the minimum congestion
required to route the demands $\bb$ in $G$.
\end{Definition}

\begin{theorem} [Theorem 1.2. from~\cite{Sherman13:arxiv}]
\label{thm:maxFlow}
There is a routine \textsc{ApproximatorMaxFlow} that, given demands $\bb$
and access to an $\alpha$-congestion-approximator $\RR$, makes
$O( \alpha^2 \log^{2}{n} \epsilon^{-3})$ iterations and returns an
$(1 + \epsilon)$-approximate flow/cut solution for these demands.
Each iteration takes $O(m)$ time, plus computing matrix-vector products
involving $\RR$ and $\RR^{T}$.
\end{theorem}

R\"{a}cke\etal~\cite{RackeST14} showed that these congestion approximators
can be efficiently computed using approximate maximum flow routines.
This result can be pharaphrased as:

\begin{theorem}
[main result of~\cite{RackeST14}]
\label{thm:cutApprox}
There is a routine \textsc{CongestionApproximator} that takes a graph $G$,
returns with high probability an $O(\log^{4}n)$-congestion-approximator $\RR$
such that matrix-vector products in $\RR$ and $\RR^T$
can be performed in $O(n)$ time.
Furthermore, this approximator is computed via
a series of approximate flow/cut solutions
with error $1 /\Theta(\log^{3}{n})$ on graphs of sizes $m_1, \ldots m_N$ such that
\[
\sum_{i = 1}^{N} m_i \leq O(m \log^{4}{n}),
\]
plus an additional running time overhead of $O(m\log^6{n})$.
\end{theorem}


This summarizes several aspects of the algorithm for constructing oblivious
routings by R\"{a}cke\etal~\cite{RackeST14}:
the fact that the oblivious routing scheme produced gives a congestion-approximator
was observed in the second paragraph of the abstract.
The approximation guarantee is from Theorem 4.1.
The invocation costs of $\RR$ and $\RR^T$ also follow from
oblivious routing scheme being a tree.

The error tolerance in the maximum flow calls of $1 /\Theta(\log^{3}{n})$
is stated in the abstract and utilized in the rebalancing step of the proof of Lemma 3.1.
Overall the algorithm performs $O(\log{n})$ levels of partition based recursion,
and the total sizes of graphs at each level is $O(m)$.
Each partition step may adjust the partition $O(\log{n})$ times using the
cut-matching game by Khandekar\etal~\cite{KhandekarRV06:journal}, which
in turn needs $O(\log^2{n})$ approximate flow/cut solutions.
Combining these bounds gives a total size of $O(m\log^4{n})$.
The running time overhead comes from applying the $O(\log^2{n})$ matchings
produced in the cut-matching game to a random vector before routing
it using approximate maximum flows.

\subsection{Ultra-Sparsifiers and Size Reductions}

Ultra-sparsifiers are controlled ways of reducing graphs to
tree-like structures.
As they involve pairs of graphs on the same vertex set,
we will use scripts to denote the graph in question in our notations.
The following construction can be obtained
from~\cite{KoutisMP10:journal} and~\cite{AbrahamN12}.

\begin{theorem}
\label{thm:ultraSparsify}
There is a routine \textsc{Ultra-Sparsify}
that takes a graph $G = (V, E_G, \uu_G)$
with $n$ vertices and $m$ edges, and any parameter $\kappa > 1$,
returns in $O(m\log{n}\log\log{n})$ time a graph $H = (V, E_H, \uu_H)$ on the same
set of vertices with $n - 1 + O(m \log^2n\log\log{n} / \kappa)$ edges such that
with high probability we have
\[
\uu_{G}(S) \approx_{\kappa} \uu_{H}(S)
\]
for all subsets of vertices $S \subseteq V$.
\end{theorem}

Since minimum cut seeks to minimize $\uu( S ) /\bb( S)$,
$\uu_G(S) \approx_{kappa} \uu_H(S)$ for all $S$ implies
$opt_G(\bb) \approx_{\kappa} opt_H(S)$, and an
$\alpha$-congestion-approximator for $H$ is also
a $\kappa \alpha$-congestion-approximator for $G$.
Note that since we only need to preserve cuts, the Spielman-Teng
construction of ultra-sparsifiers~\cite{SpielmanTengSolver:journal} with spectral
sparsifiers replaced by cut-sparsifiers~\cite{BenczurK96} also gives a similar bound.

These edge reductions are complemented by vertex reductions,
which also are crucial in algorithms using ultra-sparsifiers~\cite{SpielmanTengSolver:journal,Madry10,KoutisMP10:journal,Peng13:thesis,Sherman13,KelnerLOS14}.

\begin{lemma}[Lemma 5.8 of~\cite{Madry10:arxiv}, paraphrased]
\label{lem:reduce}
When given a graph $H$ with $n$ vertices and $m = n - 1 + m'$ edges,
we can produce a graph $H' = \textsc{Reduce}(H)$ with $O(m')$ edges
such that any $\alpha$-congestion-approximator $\RR_{H'}$ for $H'$ can
be converted into an $O(\alpha)$-congestion-approximator $\RR_H$ for $H$.
Furthermore, matrix-vector products involving $\RR_H$ or $\RR_H^T$ can be
computed by performing a single matrix-vector product in $\RR_{H'}$
or $\RR_{H'}^T$ respectively plus an overhead of $O(m)$.
\end{lemma}

Combining these two steps for edge and vertex reductions
gives our key size reduction routines:

\begin{corollary}
\label{cor:ultraReduce}
There are routines \textsc{UltraSparsifyAndReduce} and \textsc{Convert} so that
when given a graph $G$ with $n$ vertices and $m$ edges, and an approximation
factor $\kappa$,
$\textsc{UltraSparsifyAndReduce}(G, \kappa)$ produces $G'$ such that with high probability.
\begin{enumerate}
\item $G'$ has at most $O(m \log^2{n}  \log\log{n}/ \kappa)$ edges, and
\item given an $\alpha$-congestion-approximator $\RR_{G'}$ for $G'$,
$\RR_G = \textsc{Convert}(G, G', \RR')$
\begin{enumerate}
\item is an $O(\kappa \alpha)$-congestion-approximator for $G$, and
\item matrix-vector products involving $\RR_G$ or $\RR_G^T$
can each be performed using one matrix-vector product involving
$\RR_{G'}$ or $\RR_{G'}^{T}$ plus an overhead of $O(m)$. 
\end{enumerate}
\end{enumerate}
\end{corollary}

\section{Recursive Algorithm}
\label{sec:algo}

Our algorithm recursively calls the two routines
for utilizing and constructing congestion-approximators,
while reducing sizes using \textsc{UltraSparsifyAndReduce}.
Its pseudocode is given in Figure~\ref{fig:algo}

\begin{figure}[ht]

\begin{algbox}
$\ff = \textsc{RecursiveApproxMaxFlow}(G, \epsilon, \bb)$
\begin{enumerate}
\item Set $\kappa \leftarrow C \log^{6} n \log\log{n}$ for some absolute constant $C$.
\item $G' \leftarrow \textsc{UltraSparsifyAndReduce}(G, \kappa)$ \label{ln:reduce}
\item $\RR_{G'} \leftarrow \textsc{CongestionApproximator}(G')$, which in turm makes
recursive calls to \textsc{RecursiveApproxMaxFlow}. \label{ln:convenstion-approximator}
\item $\RR_{G} \leftarrow \textsc{Convert}(G, G', \RR_{G'})$ \label{ln:convert}
\item Return $\textsc{ApproximatorMaxFlow}(G, \RR_G, \epsilon)$ \label{ln:finish}.
\end{enumerate}
\end{algbox}

\caption{Recursive Algorithmm for Approximate Maximum Flow}
\label{fig:algo}
\end{figure}

We will simplify the analysis by bounding the size reductions at each
recursive call and the overall failure probability of any call.

\begin{lemma}
\label{lem:sanity}
We have $|E_{G'}| \leq O( |E_{G}| / (C \log^4{n} ))$ during each recursive call,
and with high probability all the function invocations terminate correctly,
\end{lemma}

\begin{proof}
Corollary~\ref{cor:ultraReduce} gives that the size of $G'$
is at most $O(m \log^2{n} \log\log{n} / \kappa)$, and the bound
follows from the choice of $\kappa = C \log^6{n} \log\log{n}$.

We can then follow the call structure of this algorithm and
accumulate failure probabilities.
At each step, Theorem~\ref{thm:cutApprox} gives that the
total size of the graphs recursed on is bounded by
$O(|E_{G'}| \log^4{n})  = O(|E_{G}| / C)$.
Therefore, a sufficiently large $C$ means the total number of recursive
calls is bounded by $O(m)$ with high probability.
Accumulating the failure probabilities over these steps using
the union bound then gives the overall success probability.
\end{proof}

We remark that to bound the failure probability of each recursive
call by $1 - n^{c}$, the routines also need to use $\bar{n}$, the initial
vertex count, instead of the size of the current instance.
This is a situation that occur frequently in analyses of recursive
invocations of Monte-Carlo randomized algorithms~\cite{SpielmanTengSolver:journal,Peng13:thesis,CohenKMPPRX14}.
We omit the details here due to the large number of routines used
in black-box manners.

For the rest of this proof we will assume that all black-box
invocations terminate correctly.

\begin{lemma}
\label{lem:lastStep}
On input of a graph with size $m$, with high probability the
cost of the final call to \textsc{ApproximatorMaxFlow} is at most
$O(m \log^{32}n \log^2 \log{n} \epsilon^{-3})$.
\end{lemma}

\begin{proof}
The guarantees of \textsc{CongestionApproximator} from
Theorem~\ref{thm:cutApprox} gives that
$\RR_{G'}$ is an $O(\log^{4}n)$-congestion-approximator for $G'$, and
matrix-vector products involving $\RR_{G'}$ and $\RR_{G'}$ cost $O(n)$.
Combining this with Corollary~\ref{cor:ultraReduce} gives that $\RR_G$ as returned by
$\textsc{Convert}$ on Line~\ref{ln:convert} of  \textsc{RecursiveApproxMaxFlow}
is an $O(\log^{10}n\log\log{n})$-congestion-approximator for $G$, and the
cost of matrix-vector products $\RR_G$ and $\RR_G^T$ is $O(m)$.
The running time and error guarantees then follow from Theorem~\ref{thm:maxFlow}.
\end{proof}

\begin{theorem}
With high probability, \textsc{RecursiveApproxMaxFlow}
returns an $(1 + \epsilon)$-approximate flow/cut solution in time
$O(m \log^{32}  n \log^2\log{n} \max \{\log^9{n}, \epsilon^{-3} \})$.
\end{theorem}

\begin{proof}
Theorem~\ref{thm:cutApprox} guarantees that the error parameter
in all intermediate calls is at most $\overline{\epsilon} = 1 / \Theta( \log^3{n})$.
Let the running time of $\textsc{RecursiveApproxMaxFlow}$
on a graph with $m$ edges and error $\overline{\epsilon}$ be $\mathcal{T}(m)$.
We will show by induction, or guess-and-check, that we can choose $\overline{C}$
to ensure $\mathcal{T}(m) \leq \overline{C} m \log^{41}n \log^2\log{n}$.

The base case of $m \leq \log^{10} \bar{n}$, where $\bar{n}$ is the top level
vertex count, follows from invoking existing approximate maximum flow
algorithms~\cite{ChristianoKMST10}.
For the inductive case, let the graphs that we recurse on have sizes
$m_1 \ldots m_N$.
Theorem~\ref{thm:cutApprox} and Lemma~\ref{lem:lastStep}
give the following recurrence:
\begin{eqnarray*}
\mathcal{T}(m)
& \leq &
\sum_{i = 1}^{N} \mathcal{T}\left(m_i \right) + O(m\log^6 n)
\\ & & \qquad + O(m \log^{32}n \log^2 \log{n} \overline{\epsilon}^{-3})\\
& \leq & \sum_{i = 1}^{N} \mathcal{T}\left(m_i \right) + O(m \log^{41}n \log^2 \log{n}).
\end{eqnarray*}

Since all graphs that we compute approximate maximum flows on
are subgraphs of $G'$,
Lemma~\ref{lem:sanity} allows us to invoke the inductive hypothesis, giving
\[
\mathcal{T}(m) \leq \sum_{i = 1}^{N} \overline{C} m_i \log^{41}n \log^2 \log{n}
+ O(m \log^{41}n \log^2 \log{n}).
\]

The total sizes of the graphs that we recurse on can in turn be bounded
via Theorem~\ref{thm:cutApprox} and Corollary~\ref{cor:ultraReduce}.
\[
\sum_{i = 1}^{N} m_i \leq O(\log^4{n}|E_{G'}|)
\leq O(m \log^6 {n} \log^2\log{n} / \kappa).
\]
Choosing $C$ appropriately then allows us to bound this by $m / 2$,
giving a total of
\[
\mathcal{T}(m) \leq \frac{\overline{C}}{2} m \log^{41}n \log^2 \log{n}
+ O(m \log^{41}n \log^2 \log{n}).
\]
The inductive hypothesis then follows by picking $\overline{C}$ to be twice
the constant of the trailing term.

This gives the bound of $O(m \log^{41}n \log^{2}\log{n})$ when $\epsilon$
is set to $1 / \Theta ( \log^3{n} )$.
For the general case of arbitrary $\epsilon$, Lemma~\ref{lem:lastStep} gives
a bound of $O(m \log^{32}n \log^2 \log{n} \epsilon^{-3})$.
Note that the first term is still present, since the second to last call is made
on a graph of size $m / 2$ with $\epsilon = 1 / \Theta ( \log^3{n} )$.
Summing over both terms then gives the overall runtime bound.
\end{proof}

We remark that the $\epsilon^{-3}$ term arises from a similar dependency in
the congestion-approximator based flow routine by Sherman~\cite{Sherman13},
which is stated in Theorem~\ref{thm:maxFlow}.

Invoking this algorithm in Theorem~\ref{thm:cutApprox}
also gives an $O(m \polylog(n))$ time algorithm for constructing
hierarchical tree decomposition based oblivious routing schemes.
\begin{corollary}
Given an undirected graph $G$, we can construct
in $O(m \log^{45}  n \log^2\log{n})$ time
a tree that with high probability corresponds to an
$O(\log^{4}n)$-competitive oblivious routing scheme .
\end{corollary}

\newcommand{\etalchar}[1]{$^{#1}$}


\end{document}